\documentclass{llncs}
\usepackage{amsmath,amssymb,amsmath}
\pdfoutput=1
\usepackage{ifpdf}
\usepackage{graphicx}
\usepackage{color}
\usepackage{fullpage}
\usepackage{textcomp}
\usepackage{setspace}
\usepackage{algorithmic}
\usepackage[ruled]{algorithm}
\newtheorem{observation}{Observation}

\renewcommand{\P}{{\cal{P}}}
\newcommand{\M}{{\cal{M}}}

\newcommand{\C}{{\cal{C}}}

\newcommand{\R}{{\cal{R}}}
\newcommand{\N}{{\cal{N}}}
\renewcommand{\S}{{\cal{S}}}
\renewcommand{\O}{{\cal{O}}}
\newcommand{\T}{{\cal{T}}}
\newcommand{\D}{{\cal{D}}}

\title{Querying for the Largest Empty Geometric Object in a 
Desired Location\thanks{This version is a significant update of our earlier arXiv submission arXiv:1004.0558v1. Apart from new variants studied in Sections~\ref{sec:Convex} and~\ref{sec:simple-polygon-case}, the results have been improved in Section~\ref{sec:qmec}. Part of this work was done when the 
authors were visiting  IISc, Bangalore, and  TIFR, Mumbai.
}}
\author{John Augustine\inst{1} \and Sandip Das \inst{2} \and 
Anil Maheshwari \inst{3} \and Subhas C. Nandy \inst{2} \and 
Sasanka Roy\inst{4} \and  \\ Swami Sarvattomananda \inst{5}}

\institute{School of Physical and Mathematical Sciences, 
Nanyang Technological University, Singapore
\and  Advanced Computing and Microelectronics Unit, 
Indian Statistical Institute, Kolkata, India
\and  School of Computer Science, Carleton University, 
Ottawa, Canada
\and Chennai Mathematical Institute, Chennai, India
\and  School of Mathematical Sciences, Ramakrishna 
Mission Vivekananda University, Belur, India}
\begin{document}

\maketitle

\begin{abstract}
We study new types of geometric query problems defined as follows: 
given a  geometric set $P$, preprocess it such that given a query
point $q$, the location of the largest circle that does not contain
any member of $P$, but contains $q$ can be reported efficiently. The
geometric sets we consider for $P$ are boundaries of convex and simple
polygons, and point sets. While we primarily focus on circles as the
desired shape, we also briefly discuss empty rectangles in the context
of point sets.
\end{abstract}

\vspace{-0.1in}
\section{Introduction}

Largest empty space recognition is a classical problem in computational 
geometry, and has applications in several disciplines like data-mining,
database management, VLSI design, to name a few. Here the problem is 
to identify an empty space of desired shape and maximum size in a given 
region containing a given set of obstacles. Given a set $P$ of points 
in $\mathbb{R}^2$, an {\it empty circle}, is a circle that does not 
contain any member of $P$. An empty circle is said to be a {\it maximal 
empty circle} (MEC) if it is not fully contained in any other empty circle. 
Among the MECs', the one having maximum radius is the {\it largest empty
circle}. The largest empty circle among a point set $P$ can easily be 
located by using the Voronoi diagram of $P$ in $O(n\log n)$ time \cite{T}. 
The {\it maximal empty axis-parallel rectangle} (MER) can be defined in 
a similar manner. The literature on recognizing the largest empty 
axis-parallel rectangle among obstacles has spanned over three decades 
in computational geometry. The pioneering work on this topic is by 
Namaad et al. \cite{NHL} where it is shown that the number of MERs' 
($m$) among a set of $n$ points may be $O(n^2)$ in the worst case. In 
the same paper, an algorithm for identifying the largest MER was 
proposed. The worst case time complexity of that algorithm is 
$O(min(n^2, m\log n))$. The best known result on this problem runs in 
$O(n\log^2n)$ time in the worst case. The same time complexity result 
holds for the recognition of the largest MER among a set of arbitrary 
polygonal obstacles \cite{NSB}. However, the largest MER inside an 
$n$-sided simple polygon can be identified in $O(n\log n)$ time 
\cite{BU}. The worst case time complexity for recognizing the largest 
empty rectangle of arbitrary orientation among a set of $n$ points is 
$O(n^3)$ \cite{CND}.

Although a lot of study has been made on the empty space recognition 
problem, surprisingly, the query version of the problem  has not 
received much attention to the best of our knowledge. The problem 
of finding the largest empty circle centered on a given query line 
segment has been considered in \cite{APS10}. The preprocessing time,
space and query time of the proposed algorithm are $O(n^3\log n)$,
$O(n^3)$ and $O(\log n)$, respectively. In practical applications, one
may need to locate the largest empty space of a given shape in a
desired location. For example, in the VLSI physical design, one may
need to place a large circuit component in the vicinity of some
already placed components. Such problems arise in mining large data
sets as well, where the objective is to quickly study the
characteristics (such as the area of the empty space) near 
a query point.
In this paper, we will study the query versions of the empty space 
recognition problem. If the desired object is a circle, the problem 
is referred to as {\it maximal empty circle query} (QMEC) problem, 
and if the desired object is an axis-parallel rectangle, the  problem 
is referred to as {\it maximal empty rectangle query} (QMER) problem. 
The following variations are considered. 
\begin{description}
\item[] Given a convex polygon $P$, preprocess 
it such that given a query point $q$, the largest circle inside 
$P$ that contains the query point $q$ can be identified efficiently.
\item[] Given a simple polygon $P$, preprocess it 
such that given a query point $q$, the largest circle inside $P$ 
that contains the query point $q$ can be identified efficiently.
\item[] Given a set of points $P$, preprocess it such 
that given a query point $q$, the largest circle that does not contain 
any member of $P$, but contains the query point $q$ can be identified 
efficiently.
\item[] Given a set of points $P$, preprocess it such that given 
a query point $q$, the largest rectangle that does not contain any 
member of $P$, but contains the query point $q$ can be identified 
efficiently.
\end{description}

We believe that our work motivates study of new types of geometric 
query problems and may lead to a very active research area. The main 
theme of our work is to mainly understand which problems can be 
solved in subquadratic preprocessing time and space, while ensuring 
polylogarithmic query times. Our results are summarized in Table 
\ref{TAB1}.

\begin{table} \caption{Complexity results of different variations 
of largest empty space query problem}
\begin{center}
\begin{tabular}{|l|c|c|c|c|c|} \hline 
Geometric set & Shape of & Preprocessing & Space & Query
time & Sections\\ 
& empty space   & time  &  & &\\ \hline 
Convex Polygon & circle & $O(n)$ & $O(n)$ & $O(\log n)$& \ref{sec:Convex} \\ \hline 
Simple Polygon & circle & $O(n \log^3 n)$ & $O(n \log^2 n)$ & 
$O(\log^2 n)$& \ref{sec:simple-polygon-case} \\ \hline 
Point Set & circle & $O(n^2 \log n)$ & $O(n^2)$ & $O(\log n)$ 
& \ref{sec:qmec} \\ \hline
Point Set & rectangle & $O(n^2 \log n)$ & $O(n^2 \log n)$ 
& $O(\log n)$ & \ref{sec:qmer} \\ \hline
\end{tabular}
\end{center}
\label{TAB1}
\vspace{-0.15in}
\end{table}

In the course of studying these problems, we developed two 
different ways of implementing a key data structures for storing 
$n$ circles of arbitrary sizes such that when a query point 
$q$ is given, it can report the largest of the circles that 
contains $q$. This data structure may be of independent interest 
since it may aid in several other geometric search problems.  

\vspace{-0.1in}

\section{Preliminaries: LCQ-problem} \label{sec2}
In this section, we want to build a data structure called the 
{\it largest circle query data structure} ({\tt LCQ}, in 
short) for the point location in an arrangement of circles. 
Our input is a set $\C = \{C_1, C_2, \ldots, C_n\}$ of circles 
in non-increasing order of their radii. In the preprocessing 
phase we will construct the data structure. When a query point 
$q$ is given, it must report the largest circle in $\C$ that 
contains $q$ or a null value if $q$ is not enclosed by any 
circle in $\C$. For simplicity, we assume that at most two 
circles intersect at any point on the plane.

We provide two ways of building the {\tt LCQ} data structure. 
The first method uses divide-and-conquer, leading to a solution 
that is optimized for preprocessing time. The second method 
uses a line sweeping technique similar to~\cite{SarTar86}, and 
it gives a solution with better query time. The complexity 
results are given in Table \ref{tab:LCQ} 

\begin{table}[htbp]
\caption{Table comparing the two solutions for {\tt LCQ}}
\centering
\begin{tabular}{|l|c|c|c|c|}\hline
Techniques & Preprocessing time & Space & Query time \\ 
\hline \hline
Divide-and-conquer & $O(n \log^2 n)$ & $O(n \log n)$ & 
$O(\log^2 n)$ \\ \hline
Line sweep & $O(n^2 \log n)$ & $O(n^2)$ & $O(\log n)$ 
\\ \hline
\end{tabular}
\label{tab:LCQ}
\vspace{-0.1in}
\end{table}

\vspace{-0.1in}
\subsection{A divide-and-conquer solution}

\vspace{-0.1in}
\noindent {\bf Preprocessing:}
%\subsubsection{Preprocessing}
We form a tree $\D$ of depth $O(\log n)$ as follows. Its root 
$r$ represents all the members in $\cal C$, and is attached with 
a data structure ${\tt vor}(r)$ with the circles in ${\cal C}(r) 
= {\cal C}$. The two children of root, say $\D_\ell$ and $\D_r$, 
represent the sets ${\cal C}(r_\ell) = \{C_1, C_2, \ldots, 
C_{\lfloor\frac{n}{2}\rfloor}\}$ and ${\cal C}(r_r) = 
\{C_{\lfloor\frac{n}{2}\rfloor+1}, \ldots, C_{n-1}, C_n\}$, 
respectively. These define the associated structures 
${\tt vor}(r_\ell)$ and ${\tt vor}(r_r)$ of $\D_\ell$ and $\D_r$  
respectively. The subtrees of $\D_\ell$ and $\D_r$ are defined 
recursively in the similar manner. Finally, the leaves of $\D$ 
contain $C_1, C_2, \ldots, C_n$, respectively. The tree is 
computed in a bottom-up fashion starting from the leaves. The 
task of the data structure ${\tt vor}(v)$ associated to a node 
$v$ is to efficiently report whether or not the query point $q$ 
lies inside the union of the circles in ${\cal C}(v)$ it represents. 
We will use {\it Voronoi diagram in Laguerre geometry} of the 
circles in ${\cal C}(v)$ \cite{IIM}. Each cell of this Voronoi 
diagram is a convex polygon and is associated with a circle in 
${\cal C}(v)$. The membership query is answered by performing a 
point location in the associated planar subdivision. For a node $v$, 
${\tt vor}(v)$ can be computed in $O(|C(v)|\log|C(v)|)$ time and 
membership query can be answered in $O(\log|C(v)|)$ time 
\cite{IIM}.

\noindent {\bf Query answering:}
To find the largest circle in $C_q$ containing the given query point 
$q$, we start searching from the root $r$ of $\D$. If $q$ does 
not lie in the union of circles ${\cal C}(r) = \cal C$, then $q$ 
is contained in an empty circle of size infinity. We need not
proceed further in the tree. However, if the search succeeds, we 
need to continue the search among its children. A successful search 
at a node $v$ indicates that $q$ must lie either in the union of 
circles of its left child or the right child or both. We first 
consider its left child $v_\ell$, that contains the larger 
$\frac{|{\cal C}(v)|}{2}$ circles of node $v$. We search in the 
associated structure ${\tt vor}(v_\ell)$. If the search succeeds 
(i.e., $q \in \cup_{C \in {\cal C}(v_{\ell})} C$), the search 
proceeds in the subtree rooted at $v_\ell$. However, if the search 
fails, surely $q$ lies in the union of circles ${\cal C}(v_r)$, 
and the search proceeds in the subtree rooted at $v_r$. Proceeding 
similarly, one can identify the largest circle $C_q$ containing the 
query point $q$. 

\begin{theorem} \label{lem:largest-circle}
A set $\cal C$ of $n$ circles can be preprocessed in $O(n\log^2 n)$ 
time and $O(n\log n)$ space so that {\tt LCQ} queries can be  
answered in $O(\log^2 n)$ time.
\end{theorem}

\vspace{-0.1in}
\subsection{A line sweep solution}
\vspace{-0.1in}
We assume a pair of orthogonal lines on the plane to represent the 
coordinate system. The circles are given as a set of tuples; each 
tuple representing a circle consists of the coordinates of its 
center and the radius of that circle. The ordered set of vertices 
$V=(v_1, v_2, \ldots, v_{|V|})$ of the arrangement ${\cal A}({\C})$ 
consists of the (i) leftmost and rightmost point of each circle 
in $\C$, and (ii) points in which a pair of circles intersect. We 
assume, further, that the vertices of ${\cal A}({\C})$ have unique 
$x$ coordinates, thereby allowing the elements of $V$ to be stored 
in increasing order of their $x$ coordinates. A maximal segment of 
any circle in $\C$ that does not contain a vertex is called an
{\it edge} of ${\cal A}({\C})$. We use $E$ to denote the set of all
edges of ${\cal A}({\C})$. Since we include the left and right
extremeties of a circle in the set of vertices, the edges are always
$x$-monotone. Each edge is attached with two fields $ID_1$ and $ID_2$
indicating the largest circle containing the cell above and below it
respectively. Each cell is bounded by the edges of ${\cal A}({\C})$,
and is attached with an index $ID$ indicating the largest circle
containing that cell. We compute the arrangement ${\cal A}({\C})$ as
follows:
\vspace{-0.1in}
\begin{description}
\item[Step-1] Cut each circles into pseudo-segments such that a 
pair of pseudo-segments intersect in at most one point. If any of 
these segment contains the leftmost/rightmost point of the 
corresponding circle, it is again split at that point.
\item[Step-2] Sweep a vertical line from left to right to compute 
the cells of the arrangement ${\cal A}({\C})$. We also compute the 
$ID$ field of each cell during the sweep.
\end{description}

{\bf Step-1}

Consider a circle $C_i \in \C$. Each circle $C_j \in \C$, $j\neq i$, 
creates an arc $\alpha_j^i$ along the boundary of $C_i$ that indicates 
the portion  of the boundary of $C_i$ that is inside $C_j$. In order 
to split $C_i$ into pseudo segments, we need to compute the minimum 
number of rays from the center of $C_i$ that are required to pierce 
all the arcs $\alpha_j^i$, $j=1,2, \ldots, n$, $j\neq i$. This can be 
computed using the $O(n)$ time algorithm for computing the minimum 
geometric clique cover of the circular arc graph provided the 
end-points of the circular arcs are sorted \cite{HT91}. But, we need 
to sort the end-points of the circular arcs along the boundary of 
$C_i$. Thus, the splitting of all the circles in $\C$ into pseudo 
segments need $O(n^2\log n)$ time. Tamaki and Tokuyama \cite{TT} 
showed that the number of pseudo segments may be $O(n^{\frac{5}{3}})$ 
in the worst case. Recently Aronov and Sharir \cite{AS1} 
showed that number of pseudo-segments generated from $n$ unequal 
circles is at most $n^{\frac{3}{2}+\epsilon}$, where $\epsilon$ 
can be made arbitrarily small. \\

{\bf Step-2}

In this step, we sweep a vertical line from left to right exactly as 
in \cite{JP} to compute the arrangement ${\cal A}(\C)$. During the 
sweep, four types of events may occur: (i) leftmost point of a circle 
(ii) rightmost point of a circle, (iii) an end-point of a pseudo-segment 
that is not of type (i) or type (ii), and (iv) intersection point of 
two pseudo-segments. The events of type (i), (ii) and (iii) are initially 
inserted in a heap $\cal H$. The events of type (iv) are inserted in 
$\cal H$ when these are observed during the sweep. The sweep line status  
data structure ${\cal L}$ stores the edges intersected by the sweep line 
at the current instant of time. Each pair of consecutive edges indicate a cell 
intersected by the sweep line. Each cell $\eta$ intersected by the sweep 
line (indicated by a pair of consecutive edges in the sweep line status) 
is attached with a balanced binary search tree $\tau_\eta$ containing
the radii of the circles 
overlapping on that cell. Each time an event having minimum $x$-coordinate 
is chosen from $\cal H$ for the processing. The actions taken for each 
type of event is listed below.
\begin{description}
\item[$\bullet$] While processing a type (i) event corresponding 
to a circle $C$, a new cell $\eta$ and two new edges, say $e_1$ 
and $e_2$, of ${\cal A}({\C})$ take birth. These two new consecutive 
edges are inserted in ${\cal L}$. If the new cell $\eta$ arrives 
inside an existing cell $\eta'$ in the sweep line status $\cal L$, 
then $\tau_{\eta}$, attached to the cell $\eta$, is created 
by copying $\tau_{\eta'}$ and inserting the radius of the circle $C$ 
in it. The $ID$ field attached with the cell $\eta$ is the largest 
element in $\tau_\eta$.
\item[$\bullet$] Type (ii) events are also handeled in a similar 
fashion. Here two edges are deleted from the sweep line 
status $\cal L$. Thus, a cell will also disappear from $\cal L$.
\item[$\bullet$] At a type (iii) event one edge leaves the sweep-line 
and a new edge appears on the sweep line. Here, excepting this change 
on the sweep line, no other action is needed.
\item[$\bullet$] While processing a type (iv) event, an old cell $\eta'$ 
disappears from the sweep line and a new cell $\eta$ takes birth. If the event 
is generated due to the intersection of edges $e_1$ and $e_2$ corresponding 
to the circles $C_1$ and $C_2$, then $\tau_\eta$ is obtained by doing
an $O(\log n)$ time updating of $\tau_{\eta'}$. If $\eta$ is inside
(resp. outside) of the circle $C_i$, then  the radius of $C_i$ is
inserted in (resp. deleted from) $\tau_{\eta'}$ to get $\tau_\eta$.
Finally, the largest element of $\tau_\eta$ is attached as the $ID$
field of the cell $\eta$.
\end{description}

Since the number of type (i) events is $O(n)$, and each type (i) event needs $O(n)$ time (for copying a heap for the new cell), the time needed for processing all type (i) events is $O(n^2)$ in the worst case. The number of type (ii) and type (iii) events are $O(n)$ and 
$O(n^{\frac{5}{3}})$ respectively. As mentioned above, processing each type (ii)/type (ii) event needs $O(1)$ time. The number of type (iv) events is $O(n^2)$ in the worst case, and their processing needs $O(n^2 \log n)$ time. The point location in the
arrangement ${\cal A}(\C)$ of pseudo-segments is similar to that in
the arrangement of line-segments. Using trapezoidal decomposition of
cells, one can perform the query in $O(\log n)$ time \cite{PS}. Thus,
we have the following result:

\begin{theorem}
Given a set of circles of arbitrary radii, the preprocessing time and
space complexity of the {\tt LCQ} data structure are $O(n^2 \log n)$
and $O(n^2)$ respectively, and given an arbitrary query point, the
largest circle containing it can be reported in $O(\log n)$ time.
\end{theorem}

\section{QMEC problem for convex polygon}
\label{sec:Convex}
Let $P$ be a convex polygon and $\{p_1, p_2, \ldots, p_n\}$ be its
vertices in anticlockwise order. The objective is to preprocess
$P$ such that given an arbitrary query point $q$, the largest
circle $C_q$ that contains $q$ but not intersected by the boundary
of $P$ can be 
reported efficiently. Needless to say, if $q$ lies outside or on the 
boundary of $P$, $C_q$ is a circle of infinite radius passing through 
$q$. So, the interesting problem is the case where $q$ lies inside 
$P$. Needless to mention that here $C_q$ is an MEC inside $P$. The
medial axis $M$ of $P$ is the locus of the centers of all the MECs' 
inside $P$. Let $c$ be the center of the largest MEC inside 
$P$\footnote{There can be infinitely many such MECs of equal radius 
inside $P$ in a degenerate case. But for the sake of simplicity we 
will assume that the largest MEC inside $P$ is unique} (see Figure 
\ref{fig:figure1}(a)). The medial axis consists of straight line
segments and can be viewed as a tree rooted at $c$ \cite{CSW99}. To
avoid the confusion with the vertices of the polygon, we call the
vertices of $M$ as nodes. Note that, the leaf-nodes of $M$ are the
vertices of $P$. Let us denote an MEC of $P$ centered at a point $x
\in M$ as $MEC_x$ and let $A_x$ be the area of $MEC_x$. 

\begin{figure}[thb]
\centering  \includegraphics[width=5in]{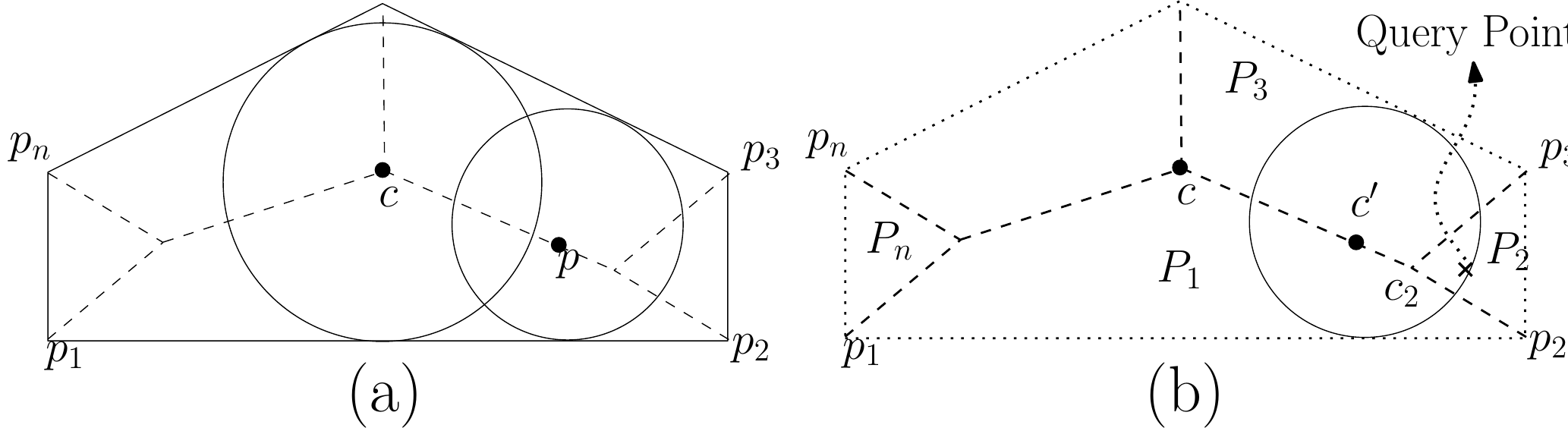}
  \caption{(a) Illustration of Observation \ref{obs1}, and  (b)
Partition of $P$.}
  \label{fig:figure1}
\end{figure}

\begin{observation} \label{obs1}
As the point $x$ moves from $c$ along the medial axis towards any
vertex $p_i \in P$ (leaf node of $M$), $A_x$ decreases
monotonically (see Figure \ref{fig:figure1}(a)).
\end{observation}
\begin{proof}
Follows from the convexity of the polygon $P$. \qed
\end{proof}

The medial axis $M$ partitions the polygon $P$ into $n$ convex
sub-polygons 
such that each sub-polygon $P_i$ consists of a polygonal edge
$p_ip_{i+1}$ 
and two convex chains of $M$, one starting at $p_i$ and other starting
at 
$p_{i+1}$ (see Figure \ref{fig:figure1}(b)). This partitioning can be 
achieved in $O(n)$ time since $M$ can be computed in linear time 
\cite{CSW99}. Moreover, $M$ can be preprocessed in $O(n)$ time so
that 
the sub-polygon containing any query point $q$ can be located in 
$O(\log n)$ time \cite{Krik83}. 

\begin{lemma}\label{llx}
The polygon $P$ can be partitioned in $O(n)$ time such that given any
arbitrary query point $q$, the edge of $M$ closest to $q$ can be 
reported in $O(\log n)$ time. 
\end{lemma}
\begin{proof}
We consider each $P_i$ separately, and compute the medial axis of
the convex chain from $p_i$ to $p_{i+1}$ (a portion of $M$). This
needs $O(\mu_i)$ time \cite{AGSS}, where $\mu_i$ is the number of
nodes in $M$ that appear as the vertices of $P_i$. Thus, for the 
entire polygon $P$, the total time complexity is $O(\sum_{i=1}^n 
\mu_i) = O(n)$, since the number of edges of $M$ is $O(n)$, and 
each edge of $M$ appears in exactly two sub-polygons. If the 
query point $q$ appears in $P_i$, then we can locate the edge of 
$M$ that is closest to $q$ in $O(\log \mu_i)$ time using point location in planar subdivision \cite{Krik83}. 
\qed
\end{proof}

Now we will describe how to solve the QMEC problem for a convex
polygon. Assume that the query point $q$ lies inside the sub-polygon
$P_i$, that is incident to the edge $p_ip_{i+1}$ of $P$. Let $c'$
denote the center of the largest MEC containing $q$. Note that, $c'$ 
will lie either on the path from $p_i$ to $c$ (denoted by $p_i \sim 
c$) or on the path from $p_{i+1}$ to $c$ ($p_{i+1} \sim c$) on $M$. 
Let us assume that $c'$ lie on the path $p_i \sim c$. We use Lemma 
\ref{llx} to identify a point $x$ on the path  $p_i \sim c$ that is 
closest to $q$ in $O(\log n)$ time. The $MEC_x$ must contain $q$. 

By Observation \ref{obs1}, we can locate $c'$ by performing a binary
search on the path $c \sim x$ that finds two consecutive nodes $v$ 
and $v'$ on the path such that $MEC_v$ encloses $q$, but 
$MEC_{v'}$ does not. In degenerate case $v$ may be $x$ and $v'$ is 
its previous node on the path $c\sim x$. Since the path lies on a 
tree representing the medial axis $M$, we can use level-ancestor 
queries \cite{BF04} for this purpose. After computing $v$ and $v'$, 
the exact location of $c'$ can be determined in $O(1)$ time. Thus, 
we have the following theorem:

\begin{theorem} \label{th-QMEC-convex}
A convex polygon on $n$-vertices can be preprocessed in $O(n)$ time
and space so that the QMEC queries can be answered in $O(\log n)$
time.
\end{theorem}

\section{QMEC problem for simple polygon}
\label{sec:simple-polygon-case}
Our approach for solving the QMEC problem in a simple polygon $P$
is based on the divide and conquer strategy, and it uses the tree
structure of the medial axis $M$. Here again the {\it leaf nodes} 
correspond to the vertices of the polygon. The {\it internal nodes}
correspond to the points on $M$ such that the MEC centered at each 
of those points touches 3 or more distinct points on the boundary 
of $P$; We use $\N$ to denote the set of internal nodes of $M$.

For the sake of simplicity in analyzing the algorithm, we assume 
that the MECs centered at the internal nodes of $M$ have distinct 
radii. A point $x \in M$, that is not a leaf, is said to be a {\em
valley point} if for a sufficiently small $\delta > 0$, the MECs
centered at points in $M$ within a distance $\delta$ from $x$ are at
least as large as $MEC_x$. We can similarly define the {\em peaks} in
$M$. We assume that the number of peaks and valley points are finite.
We use $\Phi$ and $\Theta$ to denote the set of valleys and peaks 
respectively. It is easy to observe that $\Phi \cap \N = 
\emptyset$, but $\Theta \subseteq \N$. 

Finally, we define a {\em mountain} to be a maximal subtree of 
$M$ that does not contain any valley point except its leaves. 
Notice that,
\begin{itemize}
\item[(i)] Each mountain has exactly one peak.
\item[(ii)] Each valley point is common to exactly two mountains, 
and it is a leaf for both the mountains.
\item[(iii)] If a point $x$ proceeds from a valley point of a 
mountain toward its peak, the size of $MEC_x$ increases. 
\end{itemize}
Thus, if we partition $M$ by cutting the tree at all the valley 
points, we get a set of mountains $\M = \{M_1, M_2, \ldots, 
M_{|\M|}\}$ (See Figure \ref{fig:MedialAxisPartitioning}(a)). 

\begin{figure}[t]
\vspace{-0.1in}
\begin{minipage}[c]{0.5\textwidth}
\begin{center} %
\includegraphics[height=1.2in]{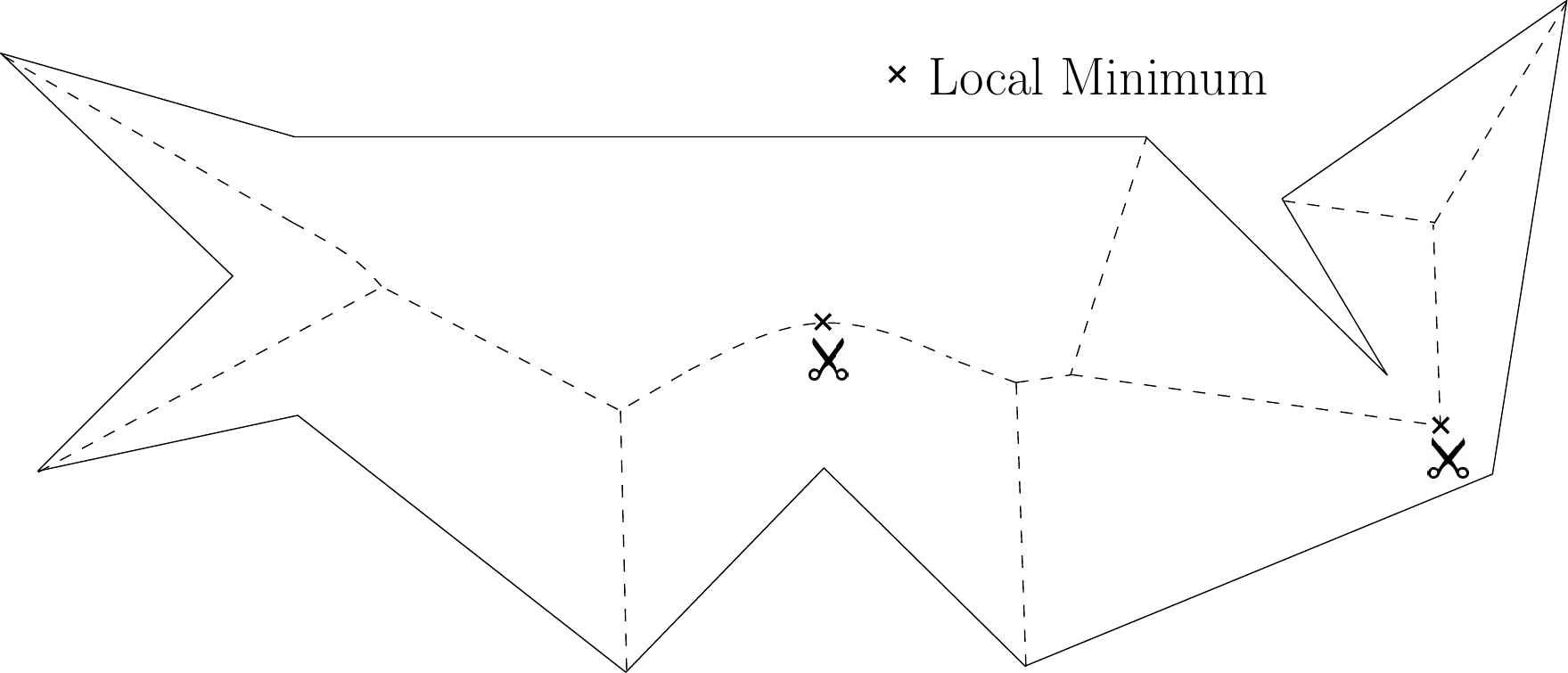}\\
(a)  
\end{center}
\end{minipage}% 
\begin{minipage}[c]{0.5\textwidth}
\begin{center} %
\includegraphics[height=1.2in]{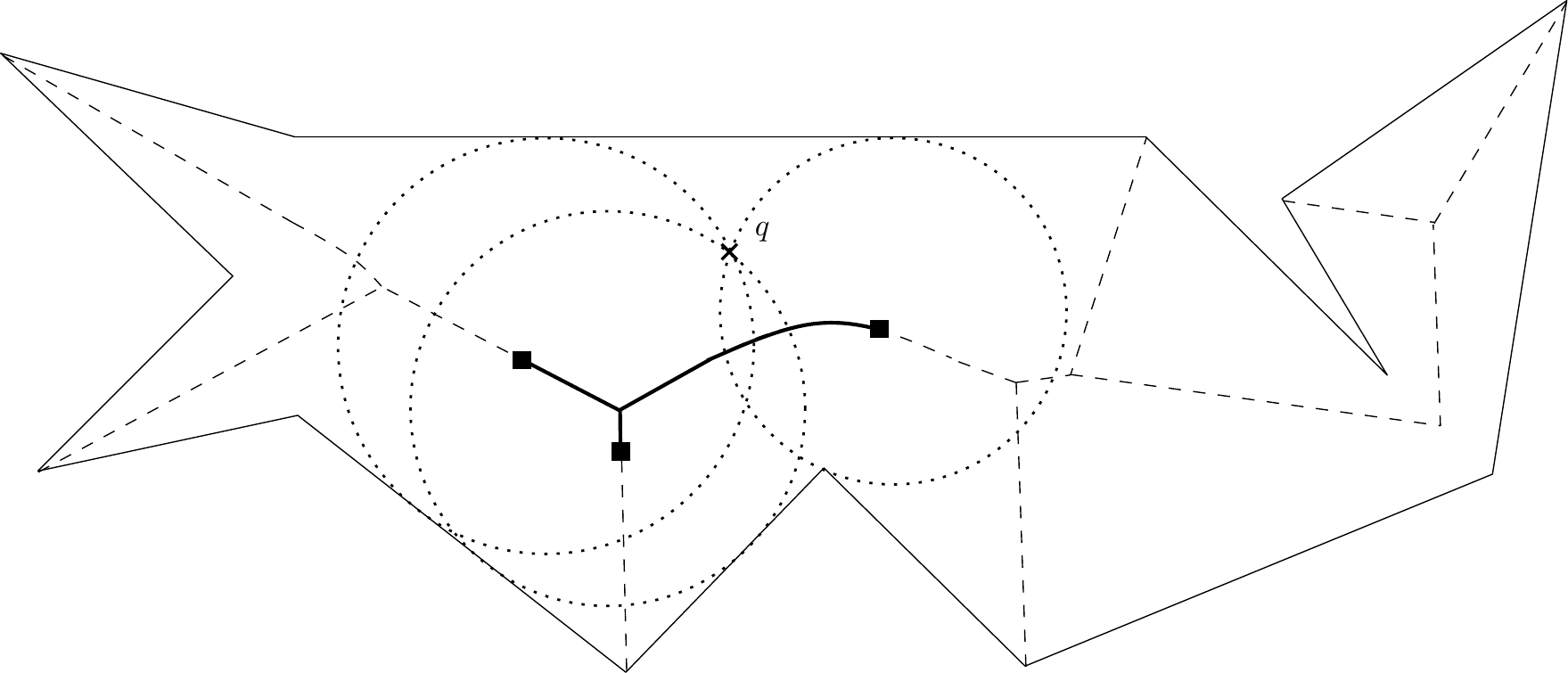}\\
(b)  
\end{center}
\end{minipage}
\vspace{-0.1in}
\caption{(a) Partitioning the medial axis $M$, and (b) The subtree $M^q$ for a query point $q$}
\vspace{-0.1in}
\label{fig:MedialAxisPartitioning}
\end{figure}

We also need to consider another way of splitting the tree $M$ 
as stated in Lemma \ref{Jordon}. This aids in designing a data 
structure $\T$ for the query algorithm.
\begin{lemma} \label{Jordon} \cite{jordon}
Every tree $T$ with $n$ nodes has at least one node $\pi$ whose 
removal splits the tree into subtrees with at most 
$\lceil\frac{n}{2}\rceil$ nodes. The node $\pi$ is called the 
centroid of $T$.
\end{lemma}
\begin{lemma} \label{l1}
If $q$ is the query point and $M^q$ is the maximal portion of $M$ 
such that MECs centered in any point on $M^q$ enclose the query 
point $q$, then $M^q$ is a connected subtree of $M$ (see Figure 
\ref{fig:MedialAxisPartitioning}(b)).
\end{lemma}

\begin{lemma}\label{l2}
If $q$ falls outside the $MEC_\pi$, then $M^q$ is contained 
entirely in one of the subtrees obtained by deleting $\pi$ 
from $M$.
\end{lemma}
\begin{proof}
Follows from the connectedness of $M^q$ (see Lemma \ref{l1}). \qed
\end{proof}
Lemmata \ref{l1} and \ref{l2} lead to the following divide and 
conquer algorithm for the {\tt QMEC} problem for a simple polygon.

In the preprocessing phase, we first compute the medial axis $M$. 
Next, we create a tree $\T$ whose root node is the centroid $\pi$ 
of $M$. The children of the root node in $\T$ are the centroids 
of the subtrees obtained by deleting $\pi$ from $M$. The process 
continues up to the leaf level. Each node $v \in {\T}$ is attached 
with $MEC_v$. Note that, the MEC attached to the root node of $\T$ 
may not be the largest MEC in $P$. 

During the query with a point $q$, we need to consider two cases: (i)
$q$ lies inside the $MEC_v$ for some vertex $v$ of the medial axis
$M$, and (ii) $q$ does not lie in the MEC of any vertex of the medial
axis. We describe the method of computing $C_q$ in Case (i). In Case
(ii) (i.e., where Case (i) fails), then we identify the mountain $M_i$
in which $q$ lies. Next, we find $C_q$ in $M_i$ using the same method
as in the Convex polygon Case, described in Section \ref{sec:Convex}.

The method of solving Case (i) is as follows. We test whether $q$ lies
in the MEC attached to the root node of $\T$. If so, we report the
largest MEC $C_q$ using a data structure {\it Query-in-Circle} (or
{\tt QiC} in short), described below. If $q$ does not lie in the MEC
corresponding to the root node, then by Lemma \ref{l2}, we need to
search one of the subtrees of the root node. The search process
continues until a node $v'$ of $\T$ is identified such that $MEC_{v'}$
contains $q$.

During the search in the tree $\T$, suppose we have identified 
a node $v$ such that $q$ lies inside $MEC_v$. Thus, $v$ lies on the 
subtree $M^q$. Here two important things need to be noted: (i) $C_q$ 
may not be equal to $MEC_v$; it may be some other MEC of larger area 
centered on $M^q$, and (ii) $M^q$ may consist of several mountains, 
The task of the {\tt QiC} data structure attached to a node $v$ of 
$\T$ is to identify the appropriate mountains in $M^q$ for searching 
the center of $C_q$. We also need another data structure, called 
{\it MEC-in-Mountain} (or {\tt MiM} in short) that can report the 
largest MEC containing $q$ with center on a given mountain $M_i 
\in {\cal M}$, provided $M_i \cap M^q \neq \emptyset$. We now explain 
{\tt MiM} and {\tt QiC} procedures in detail, and then the divide 
and conquer procedure.

\subsection{{\tt MiM} query}\label{subsec:MiM}
Here we are given the polygon $P$ and a mountain $M_i$; we need to 
report the largest MEC centered at a point on $M_i$ provided $M_i 
\cap M^q \neq \emptyset$. Note that, if the center moves from any 
point $x \in M_i \cap M^q$ to the peak of $M_i$, the MECs' are 
strictly increasing. Thus, we can apply the algorithm proposed in 
Section~\ref{sec:Convex} to identify the largest $MEC$ containing 
$q$, and centered on $M_i \cap M^q$. The preprocessing time and 
space complexities are both $O(|M_i|)$, and the query time is 
$O(\log |M_i|)$, where $|M_i|$ denotes the number of sides of the simple polygon that induces the edges of $M_i$.

\vspace{-0.2in}
\subsection{{\tt QiC} query}\label{subsec:QiC}
Here we want to solve a subproblem in which we know that the query
point $q$ falls inside an MEC centered at a given point $v \in M$. We
are to preprocess this information. In the query phase, given a query
point $q\in MEC_v$, we are required to report $C_q$, the largest MEC
containing $q$. This problem is quite challenging since the locus
$M_v^*$ of the center of $C_q$ for possible choices of $q$ satisfying
above, is a subtree of $M$, and it may span several mountains. 

During the breadth-first search in $\T$, suppose we have already 
identified a vertex $v$ in $\T$ such that $MEC_v$ contains the
query point $q$. But $C_q$ may be some other MEC of larger area. We
need to identify $C_q$. By Lemma \ref{l1}, both the center of $C_q$
and the node $v$ of $\T$ are guaranteed to be on $M^q$. 

Let $\R_v$ be the set of radii of MECs centered at the internal 
nodes of the subtree $M_v^*$ rooted at $v$, sorted in increasing
order. 

\begin{definition} \label{def1}
An MEC $C$ is called a guiding MEC corresponding to a node $v$ 
of $M_v^*$ if 
\vspace{-0.1in}
\begin{itemize}
\item[$\bullet$] its radius is in $\R_v$,
\item[$\bullet$] every MEC in the path from $v$ to the center of $C$
(both inclusive) is no larger than $C$, and
\item[$\bullet$] $C$ overlaps with $MEC_v$. 
\end{itemize}
\end{definition}
Let $\S$ be the set of all guiding MECs of the node $v$. Note that, a
member in $\S$ may be centered at the nodes as well as on the edges on
$M_v^*$. 

\vspace{-0.15in}
\subsubsection{Preprocessing steps in {\tt QiC}}
We perform the following steps in the preprocessing phase to compute 
$\S$ attached to a node $v$ of $M$. Let $M_v^*$ be the subtree of $M$
attached to node $v$.
\begin{enumerate}
\item Perform a breadth first search in $M_v^*$ starting at $v$, 
and it recursively proceeds as follows:

At each step (at a node $v' \in M_v^*$), if $MEC_{v'}$ does not 
overlap with $MEC_v$, the recursion stops along that path; otherwise, 
two distinct cases need to be considered. We check whether the radii 
of $MEC_{v'}$ and $MEC_{v''}$ are consecutive elements in $\R_v$,
where $v''$ is the predecessor node of $v'$ in $M_v^*$.
\begin{itemize}
\item[$\bullet$]  If so, put $MEC_{v'}$ in $\S$ and recursively 
explore all the paths incident at $v'$. 
\item[$\bullet$] Otherwise, compute all the MECs with center on the 
line segment $(v,v')$ whose radius matches with the elements in the 
array $\R_v$, put them in $\S$, insert those points on $(v,v')$ as
the (dummy) nodes in the tree $M$, and then recursively explore all
the paths incident at $v'$ in $M_v^*$.  
\end{itemize}
\item Attach the mountain-id with each $C \in \S$. This is available
while performing the breath-first search. This will allow us to invoke
the {\tt MiM} query for a  particular mountain.
\item Attach each circle in $\S$ with the corresponding mountain in
$M$. 
\item Create a {\tt LCQ} data structure with the circles in $\S$, and
attach it with node $v$.
\end{enumerate}

\begin{lemma}\label{lem:BoundingS} For any $r \in \R_v$, the number of
circles in $\S$ of radius $r$ attached with node $v$ is bounded by a
constant. Furthermore, $\S$ can be computed in $O(|R_v|)$ time.
\end{lemma}

\begin{figure}[t]
\vspace{-0.1in}
\begin{minipage}[c]{0.5\textwidth}
\begin{center} %
\includegraphics[height=2.4in]{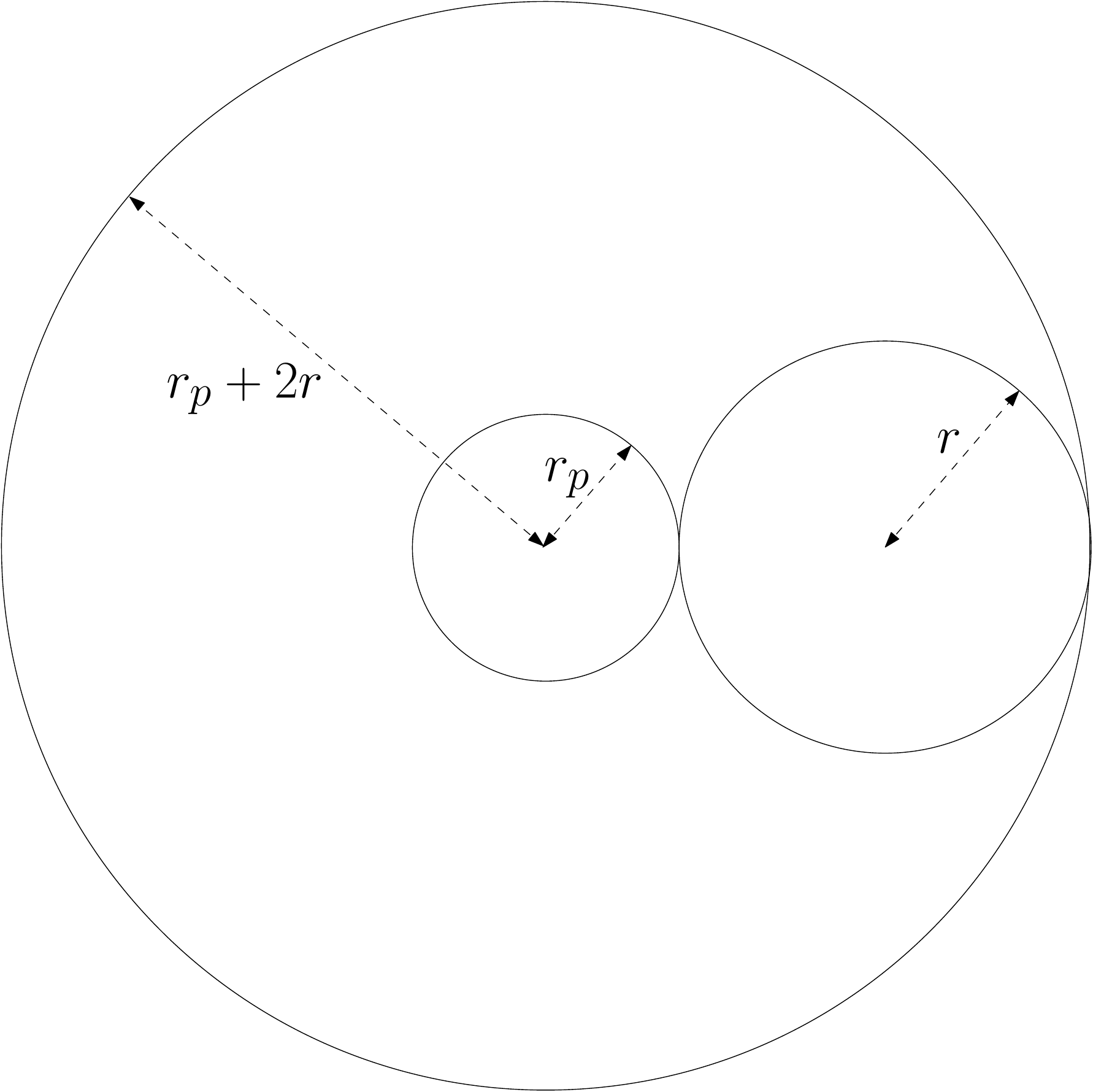}\\
(a)  
\end{center}
\end{minipage}% 
\begin{minipage}[c]{0.5\textwidth}
\begin{center} %
\includegraphics[height=2.4in]
    {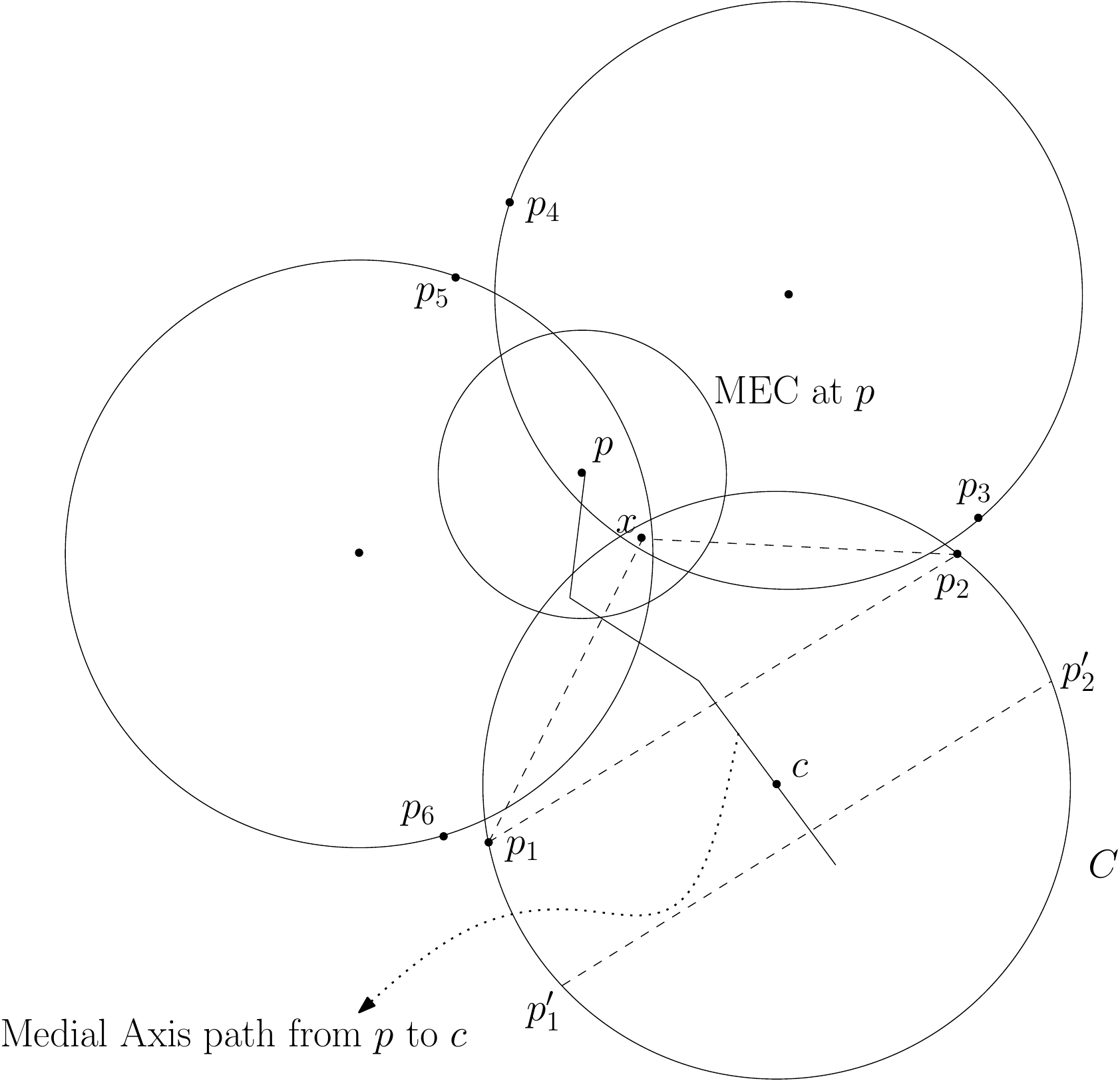}\\
(b)  
\end{center}
\end{minipage}
\vspace{-0.1in}
\caption{(a) Bounding $|\S|$, and (b) Illustration on the number of
MECs in $\S_r$ that enclose a point $x$}
\vspace{-0.1in}
\label{fig:BoundingS_r}
\end{figure}

\begin{proof}
Consider any $r \in \R_v$. Let $\S_r$ be the MECs of radius $r$ in 
$\S$. It suffices to show that $|\S_r|$ is bounded by a constant. 
For convenience, let us assume that $\S_r$ does not contain a MEC 
centered at a node of $M$. This will not affect us because we have 
assumed that MECs centered at nodes of $M$ have distinct radii, so 
at most one MEC in $\S_r$ can be centered at a node.

Let $r_v$ be the radius of $MEC_v$ of the root node $v$ of $M_v^*$.
Clearly, $r_v \le r$ (by Definition \ref{def1}). Also recall that
every MEC in $\S$ must (at least tangentially) intersect $MEC_v$.
See Figure \ref{fig:BoundingS_r}(a) for an illustration. Therefore,
every MEC in $\S_r$ must lie entirely within a circle $\chi$ of 
radius $r_v + 2r$ centered at $v$. Thus, we need to prove that the 
number of guiding circles of radius $r$ at node $v$ inside $\chi$ is
bounded by a constant.

Let us consider a point $x \in P$ . Let $\S_r^x \subseteq \S_r$ 
be a set of MECs that enclose $x$. Let $C$ be any MEC in $\S_r^x$ 
and $c$ be its center. Let $p_1$ and $p_2$ be the two points at 
which $C$ touches the boundary of the polygon $P$. The chord 
$[p_1,p_2]$ must intersect the medial axis (see Figure 
\ref{fig:BoundingS_r}(b)). Note that, the points $p$ and $c$ lie 
in the two different sides of $[p_1,p_2]$. On the contrary, if $p$ 
and $c$ lie in the same side of $[p_1',p_2']$, where $p_1'$ and 
$p_2'$ are the points of contact of the said MEC and the polygon 
$P$, then we can increase the size of the MEC by moving its center 
$c$ towards $p$ along the medial axis (see Figure
\ref{fig:BoundingS_r}(b)). Thus, $C \not\in \S_r$. Thus, we have
$\angle p_1 x p_2 \ge \pi/2$. These angles subtended by the MECs in
$\S_r^x$ are disjoint implying that $|\S_r^x| \le 4$. In other words,
any point inside the circle $\chi$ can be enclosed by at most four
different circles from $\S_r$.  We need to compute $|\S_r|$. Let us
consider a function $f(x) =$ number of circles in $\S_r$ that overlaps
at the point $x$, $x \in \chi$. $f(x) \leq 4$ for all $x \in \chi$.
The total number of circles in $\S_r$ can be obtained as follows:

Total area of circles in $\S_r \leq   \int_{(x,y) \in \chi} f(x) ~dx
~dy \le 4 \pi(r_v + 2r)^2$. \\ 

Therefore, $|\S_r|  \leq  \frac{4 \pi(r_v + 2r)^2}{ \pi r^2} \leq
\frac{4 \pi (3r)^2}{ \pi r^2} = 36$.

Thus, the first part of the lemma is proved.

The time complexity follows from the fact that the breadth first
search in $M_v^*$ needs $O(|M_v^*|)$ time. The time for computing the
members in $\S$ is $\sum_{r \in \R_v} |\S_r| = O(|\R_v|)$ (by the 
first part of this lemma). \qed
\end{proof}

\subsubsection{Query algorithm in {\tt QiC}}
Given a query point $q$, we first traverse the tree $\T$ to identify 
a node $v$ such that $q \in MEC_v$. Note that, the MECs at the nodes 
of the subtree rooted at $v$ may contain $q$; but the MECs 
corresponding to all other nodes in $\T$ will not contain $q$ (see 
Lemma \ref{l2}). Let $\S$ be the guiding circles attached with node
$v$, $\rho \in \R_v$ be the radius of the largest guiding circle in
$\S$ that contains $q$, and $\S_q$ be a subset of $\S$ that has radius
$\rho$ and contains $q$. $\S_q$ can be obtained from the {\tt LCQ}
data structure attached with node $v$. By Lemma~\ref{lem:BoundingS},
we have $|S_q| \leq 36$. In order to report $C_q$, we need the
following: 
\begin{description}
\item[Step 1:] an algorithm to identify the mountain associated with
each circle in $\S_q$,
\item[Step 2:] to locate the largest MEC containing $q$ in each of
these mountains using the {\tt MiM} query algorithm,
\item[Step 3:] to report the largest one among the MECs' obtained in
Step 2 as $C_q$. 
\end{description}
We first devise an algorithm for Step 1. The necessary algorithm for
Step 2 is already available in Subsection \ref{subsec:MiM}. We then
prove the necessary result to ensure the statement stated in Step 3.
\\

{\bf Algorithm for Step 1:} 

Consider a path $\Pi$ from $v$ to a leaf of $M_v^*$, and observe the
size of the MECs'. Figure \ref{fig:CorrectnessLemma} demonstrates a
curve $f(t)$ where $t$ denotes the distance of a point from $v$ on 
the path $\Pi$, and $f(t)$ denotes the radius of the MEC centered 
at that point. The guiding circles along the path $\Pi$ correspond 
to a subsequence of vertices along that path whose corresponding 
MECs' are increasing in size. 

\begin{lemma}\label{junk}
The guiding circles along a path $\Pi$ from $v$ to a leaf of $M_v^*$ 
containing the query point $q$ appear consecutively along $\Pi$. 
\end{lemma}
\begin{proof}
Follows from the connectedness of $M^q$ (see Lemma \ref{l1}). \qed
\end{proof}
Consider the MECs' attached to the nodes in $\Pi$. Let $v'$ be such a 
node whose corresponding MEC is largest among those containing $q$.
Let $M'$ ($\in \M$) be the mountain in which $v'$ lies. Here two cases
need to be considered: (i) $v'$ is the peak of $M'$, and (ii) $v'$ is
not the peak of $M'$. In Case (i), we have already got the largest MEC
centered on the path $\Pi$ and containing $q$. In Case (ii), we need
to invoke {\tt MiM} query algorithm to find the largest MEC centered
on the mountain $M'$.\\

{\bf Correctness of {\tt QiC}}

\begin{lemma}\label{lem:correctness}
At least one of the circles in $\S_q$ is centered in the mountain in
which $C_q$ is centered. 
\end{lemma}
\begin{figure}[t]
\vspace{-0.2in}
\begin{center}
\includegraphics[width=0.5\columnwidth]{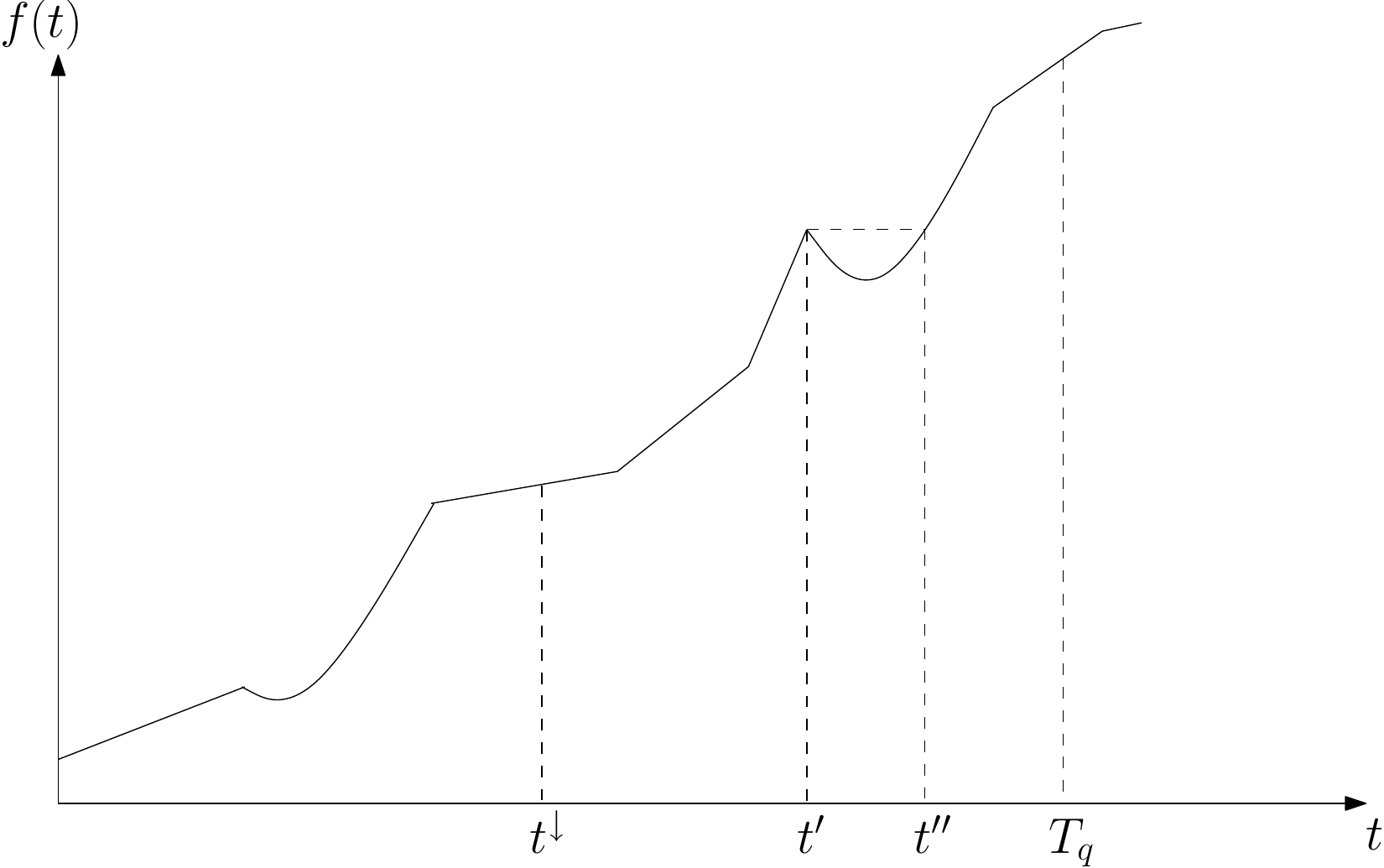}
\end{center}
\vspace{-0.1in}
\caption{Proof of Lemma \ref{lem:correctness}}
\label{fig:CorrectnessLemma}
\vspace{-0.2in}
\end{figure}

\begin{proof}
Since $M^q$ is a continuous subtree of $M$, if we explore all the
paths in $M_v^*$ from node $v$ towards its leaves, $c_q$ is reached
in one of such paths, say $\Pi$, and $v'$ be a node on $\Pi$ such 
that the guiding circle $MEC_{v'}$ is largest among those which 
contain $q$. Note that, any point on the path $\Pi$ closer to $v$ 
than $v'$ can not be the center of a larger MEC (see Definition 
\ref{def1}). Let the center $c_q$ of $C_q$ be a point on $\Pi$ that 
is in a different mountain to the right of $v'$. Here again two 
situations need to be considered: (i) the function $f(t)$ increases 
monotonically from $v'$ to $c_q$, and (ii) the function $f(t)$ from 
$v'$ to $c_q$ is not monotonic. In Case (i) $v'$ and $c_q$ lie in 
the same mountain. In Case (ii), between $v'$ and $c_q$ there is a 
point $\alpha$ on the path $\Pi$, such that the radius of 
$MEC_\alpha$ is less than that of $MEC_{v'}$. Also, there exists 
another point $\beta$ on the path $\Pi$ between $\alpha$ and $c_q$ 
such that the radius of $MEC_\beta$ is equal to that of $MEC_{v'}$. 
Since the radius of $MEC_\beta$ matches with an element of $\R$, 
$MEC_\beta$ is also a guiding circle. Moreover, from the continuity 
of $M^q$, the MEC centered at $\beta$ must contain $q$. So, if $c_q$ 
does not lie in the mountain of $v'$, it must lie in the mountain  
containing $\beta$. Thus, the lemma follows. \qed
\end{proof}

\begin{lemma}\label{lem:QiC}
The preprocessing time and space complexities for the {\tt QiC} query
are $O(|M| \log^2 |M|)$ and  $O(|M| \log |M|)$ respectively. Queries
can be answered in  $O(\log^2 |M|)$ time.
\end{lemma}

\begin{proof} Compuing $\S$ requires $O(|M_v^*| \log |M_v^*|)$ time
because we need to sort the elements in $R_v$. The members in $\S$ can
be stored in {\tt LCQ} data structure in $O(|\R_v| \log^2 |\R_v|)$
time and $O(|\R_v| \log |\R_v|)$ space (see Theorem
\ref{lem:largest-circle}) and queries can be answered in $\log^
2 |M_v^*|$.

In the query phase with a query point $q$, we identify a constant
number of guiding circles $\S_q$ attached to node $v$ that contains
$q$. Next, we call {\tt MiM} queries in their associated mountains;
this takes $\log |M_v^*|$ time (see Theorem \ref{th-QMEC-convex}). 
By Lemma~\ref{lem:correctness}, the result of one of the {\tt 
MiM} queries will be the largest MEC with center on $M_v^*$ that 
contains $q$. Thus the query time complexity follows. \qed
\end{proof}

\subsection{Query algorithm for finding $C_q$}
Before we start the divide and conquer, we compute the set $\M$ of
mountains and preprocess each of them for {\tt MiM} query. Since $\M$ 
is a partition of the medial axis $M$, all the mountains can be
preprocessed for {\tt MiM} query in $O(|M|)$ time. In Lemma
\ref{lem:QiC}, it is shown that the total preprocessing time needed
for the {\tt QiC} queries at every node of $\T$  is $O(|M| \log^2
|M|)$ using $O(|M| \log |M|)$ space. 

In the query phase, we start at the root level of $\T$ and check if
$q$ falls inside the MEC centered at the root. If yes, we find $C_q$
using the query algorithm for {\tt QiC}. This takes $O(\log^2 n)$ 
time (see Lemma \ref{lem:QiC}). Otherwise, we proceed in the 
appropriate subtree of the root whose corresponding sub-polygon 
contains $q$. In order to choose this sub-polygon, we need another 
data structure as stated below. 
\begin{figure}
\vspace{-0.1in}
\centering
\includegraphics[width=0.6\columnwidth]{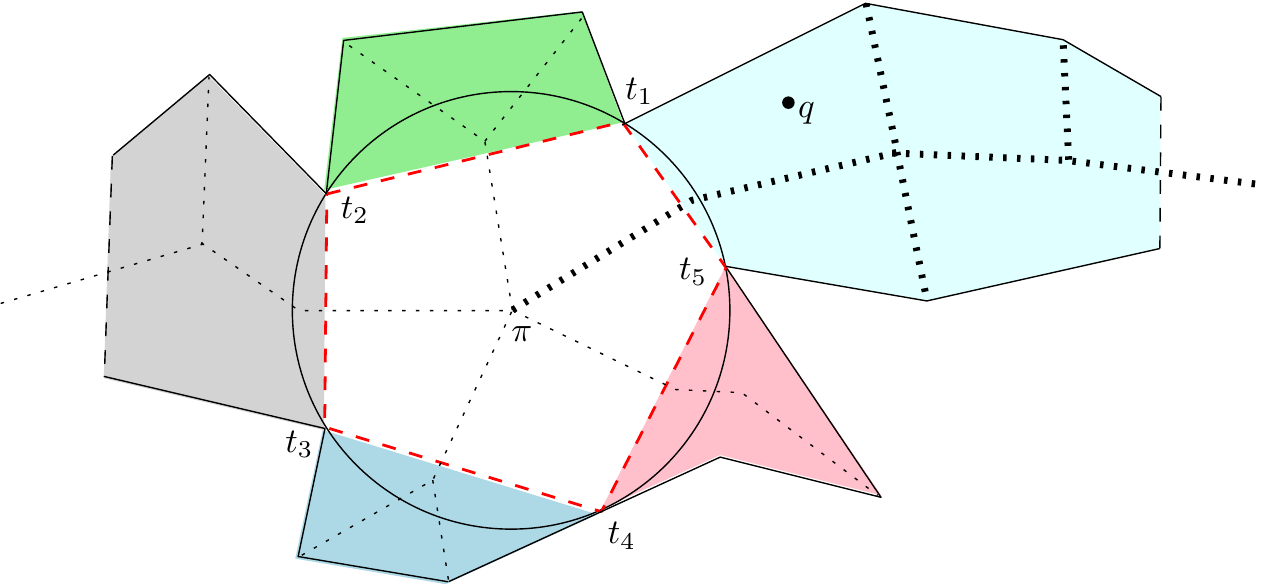}
\caption{The divide and conquer search structure}
\vspace{-0.2in}
\label{recursion}
\end{figure}

Recall that, the medial axis $M$ partitions $P$ into $n$ cells.
Let $\pi$ be the centroid node that corresponds to the root of $\T$.
Let $\pi$ have $k$ children. In other words, if we consider the
$MEC_\pi$, it touches $P$ at $k$ different points. This gives birth 
to $k$ sub-polygons (as illustrated in Figure \ref{recursion} with
$k=5$). The centroid of each sub-polygon is a child of $\pi$. We
attach a {\it first-level-tag} $i$ with each cell in the $i$-th
sub-polygon, for $i = 1,2,\ldots, k$. Next, we consider the children
of $\pi$ (the nodes in the second level of $\T$) in a breadth first
manner. For each sub-polygon, consider its centroid. The MEC of that
node again partition that sub-polygon into further parts. We attach a
{\it second-level-tag} to each cell of that sub-polygon as we did for
the root. After considering all the children of $\pi$, we go to the
third level, and do the same for attaching the {\it third-level-tag}
to the partitions of $P$. Since the number of levels of $\T$ is
$O(\log n)$ in the worst case, a cell of $P$ may get $O(\log n)$ tags.
Thus each cell is attached with an array {\it TAG} of size $O(\log n)$
containing tags of $O(\log n)$ levels. This needs $O(n\log n)$ time
and space in the worst case.

By point location in the planar subdivision of $P$, we know in which 
partition $Q$ of $P$ the query point $q$ lies. While searching in the 
tree $\T$, if $q$ does not lie in $MEC_u$ of a node $u$ in the $i$-th
level, we choose the appropriate subtree of $u$ by observing the
$i$-th entry of the array {\it TAG} attached to the partition $Q$, 
and proceed in that direction. Thus, the overall query time
complexity includes (i) $O(\log n)$ for the point location in the
subdivision of $P$, (ii) $O(\log n)$ time for traversal in $\T$, (iii)
$O(\log^2 n+K\log n)$ time for identifying the largest guiding circles
attached to node $v$ (in its $LCQ$ data structure) if node $v$ is
observed first during traversal of $\T$, such that $MEC_v$ contains
$q$, and the {\tt MiM} queries for $K$ mountains if $K$ circles are
output of step (iii). In Lemma \ref{lem:BoundingS}, it is proved that
$K$ is bounded by a constant. Thus, we have the following theorem.

\begin{theorem} \label{thm:SimplePolygon}
A simple polygon can be preprocessed in $O(n \log^3 n)$ time using 
$(n \log^2 n)$ space and the QMEC queries can be answered in 
$O(\log^2 n)$ time.
\end{theorem}

\section{QMEC for Point Set}
\label{sec:qmec}
The input consists of a set of points $P=\{p_1, p_2, \ldots, p_n\}$ 
in $\mathbb{R}^2$. The objective is to preprocess $P$ such that 
given any arbitrary query point $q$, the largest circle $C_q$ that
does not contain any point in $P$ but contains $q$, can be reported 
efficiently. Observe that, if $q$ lies outside or on the boundary 
of the convex hull of $P$, we can draw a circle of infinite radius 
passing through $q$. So, we shall consider the case where $q$ 
lies in the proper interior of the convex hull of $P$. 

An MEC centered at a {\it Voronoi vertex} touches at least three
points from $P$. We assume that the MECs centered at {\it Voronoi
vertices} are of distinct sizes. For our purpose, we also compute 
some {\em artificial vertices}, one on each Voronoi edge that is 
a half line. We must compute these artificial vertices carefully 
to ensure that the following conditions hold. 
\begin{enumerate}
\item Every MEC centered at an artificial vertex must be larger than 
MECs centered at Voronoi vertices, and
\item the MECs centered at artificial vertices should not overlap
pairwise within the convex hull of $P$. Surely, they overlap outside
the convex hull of $P$.
\end{enumerate}
The second condition ensures that there exists no query point $q$
which can be enclosed by more than one MEC centered at artificial
vertices. This second condition makes the choice of artificial
vertices somewhat tricky, but it is a simple exercise to see that we
can choose the artificial vertices in $O(n^2)$ time. We use the
unqualified term {\em vertex} to refer either to a  Voronoi vertex or
an artificial vertex.

Now, consider the planar graph with both the Voronoi vertices and
the artificial vertices. Let $v$ be a Voronoi vertex and let
$MEC_v$ be the MEC centered at $v$. A path $\P_v = (v^1=v, v^2,
\ldots, v^k)$ from $v$ in the graph is said to be a rising path with
respect to $v$ if 

\begin{itemize}
\item[$\bullet$] MECs centered at vertices other than $v^1$ and $v^k$
are strictly smaller than $MEC_v$, and
\item[$\bullet$] $MEC_{v^k}$ is strictly larger than $MEC_v$. Note
that $v^k$ may be an artificial vertex, but the other vertices in the
path $\P_v$ are surely Voronoi vertices.
\end{itemize}

The last edge $(v^{k-1}, v^k)$ is called a {\em rising edge} with
respect to the vertex $v$. Since $MEC_{v^{k-1}}$ is smaller than
$MEC_v$, but $MEC_{v^k}$ is larger than $MEC_v$, there is exactly one
MEC centered on the edge $(v^{k-1},v^k)$ that equals the size of
$MEC_v$. Let us denote this MEC by $MEC_{\P_v}$. If $MEC_{\P_v}$
overlaps $MEC_v$, then the rising edge $(v^{k-1}, v^k)$ is called an
{\em overlapping edge} with respect to vertex $v$. Let $\O_v$ be the
set of overlapping edges with respect to $v$. Observe that $\O_v$, for
a given vertex $v$, can be computed in $O(n)$ time via a breadth first
search from $v$. The preprocessing and query procedures are given in
Procedures~\ref{alg:preprocessing} and \ref{alg:query}.

\floatname{algorithm}{Procedure}
\begin{algorithm}[h!]
\caption{Preprocessing steps}
\label{alg:preprocessing}
\begin{algorithmic}[1]

\STATE{\bf INPUT:} Set of points $P$ in $\mathbb{R}^2$.
\STATE Compute the MECs centered at (both Voronoi and artificial) 
vertices, and store them in {\tt LCQ} data structure. 
\STATE For each vertex $v$, compute $\O_v$ using a breadth first
search. \label{lno:bfs}

\end{algorithmic}
\end{algorithm}

\begin{algorithm}[h!]
\caption{Query steps}
\label{alg:query}
\begin{algorithmic}[1]

\STATE{\bf INPUT:} a query point $q$ along with the {\tt LCQ} data
structure containing MECs centered at vertices, and $\O_v$ for every
internal vertex $v$ in the the Voronoi diagram of $P$. 

\STATE Find the largest MEC $C^q$ in the {\tt LCQ} data structure
containing $q$. Let $v^q$ be the center of $C^q$. 

\IF{$v^q$ is an artificial vertex}

\STATE Report $C_q=C^q$ as the largest circle containing $q$ that is
centered on the edge containing $v^q$. 

\STATE Exit

\ENDIF

\STATE $C \leftarrow C^q$

\FORALL{edges $e \in \O_{v^q}$} \label{lno:for}

\IF{MECs centered on $e$ do not enclose $q$}
\STATE Continue to next edge in $O_{v^q}$
\ENDIF

\STATE Let $C^e$  be the largest circle centered on $e$ that encloses
$q$. \label{lno:update} 

\STATE $C \leftarrow \max(C, C^e)$

%\IF{$C$ is smaller than $C^e$} 
%
%\STATE $C \leftarrow C^e$
%
%\ENDIF

\ENDFOR

\STATE Report $C$
\end{algorithmic}
\vspace{-0.1in}
\end{algorithm}

We are now left with showing that Procedures~\ref{alg:preprocessing}
and \ref{alg:query} are correct and bound their complexities. We
address the latter first. We begin with a lemma. which can be proved
essentially using the proof of  Lemma~\ref{lem:BoundingS}.  

\begin{lemma} \label{lem:BoundingS2}
For any internal vertex $v$ in the Voronoi diagram of $P$, $|\O_v|$ is
bounded by a constant. 
\end{lemma}
\begin{proof}
Let $MEC_v$ be the circle centered on $v$. Consider any overlapping
edge $e=(v_1, v_2)$ in $\O_v$. Assume without loss of generality that
the MEC at $v_2$ is larger than the MEC at $v_1$. By definition, there
is a  point $v'$ on $e$ such that $MEC_{v'}$ has the same radius as
$MEC_v$ and that $MEC_{v'}$ intersects $MEC_v$. Let $p_1$ and $p_2$ be
the two points in $P$ that touch the MEC at $v'$. The chord $p_1 p_2$
intersects the edge $e$ somewhere between $v_1$ and $v'$ (see Figure
\ref{fig:BoundingS_r}). Otherwise, the MEC at $v'$ will not be the
first MEC from $v_1$ to $v_2$ that equals $MEC_v$ in size. Therefore,
we can use the same idea from Lemma~\ref{lem:BoundingS} to bound the
number of overlapping edges. \qed
\end{proof}

\begin{lemma}\label{lem:PointSetBounds}
Given that we can construct the {\tt LCQ} data structure for $n$
circles in $O(p(n))$ preprocessing time with a space complexity of
$O(s(n))$ and queries answered in $O(q(n))$ time, Procedure
\ref{alg:preprocessing} (for preprocessing) takes $O(p(n) +n^2)$ time
and $O(s(n) + n)$ space, and Procedure~\ref{alg:query} (for query
answering) takes $O(q(n))$ time. 
\end{lemma}
\begin{proof}

The complexity bounds for {\tt LCQ} data structure are added for the
obvious reason that we use the {\tt LCQ} data structure for creating
and storing $O(n)$ MECs centered at the internal vertices of the
Voronoi diagram.  Line number~\ref{lno:bfs} of Procedure
\ref{alg:preprocessing} performs $O(n)$ breadth first searches, hence
we added an $O(n^2)$ term to the preprocessing time. As a consequence
of Lemma~\ref{lem:BoundingS2}, our space requirements is limited to
$O(n)$ and, more importantly, the query time does not incur anything
more than $q(n)$. \qed
\end{proof}
\begin{lemma}\label{lem:cycle}
Consider any cycle $H$ in the Voronoi diagram of $P$. Let $C_H$ be any
MEC centered at some point on $H$. Then, there exists another MEC
$C'_H$ centered at some other point on $H$ that does not properly
overlap $C_H$. 
\end{lemma}
\begin{proof}
Clearly, any cycle in the Voronoi diagram of $P$ must contain at 
least one point from $P$ inside it. Let $p \in P$ be such a  point
that lies inside the cycle $H$ (see Figure~\ref{fig:cycle}). Let $C_H$
be any MEC centered at some point on $H$; $c_H$ be its center.
Consider the line connecting $c_H$ and $p$. It intersects $H$ at
another point $c'_H$. It is easy to see that the MEC $C_H'$, centered
at $c'_H$, will not properly overlap with $C_H$. Because, in that case
$p$ will be properly contained within $C_H$ and $C'_H$. 
\begin{figure}
\centering
\includegraphics[width=0.40\textwidth]{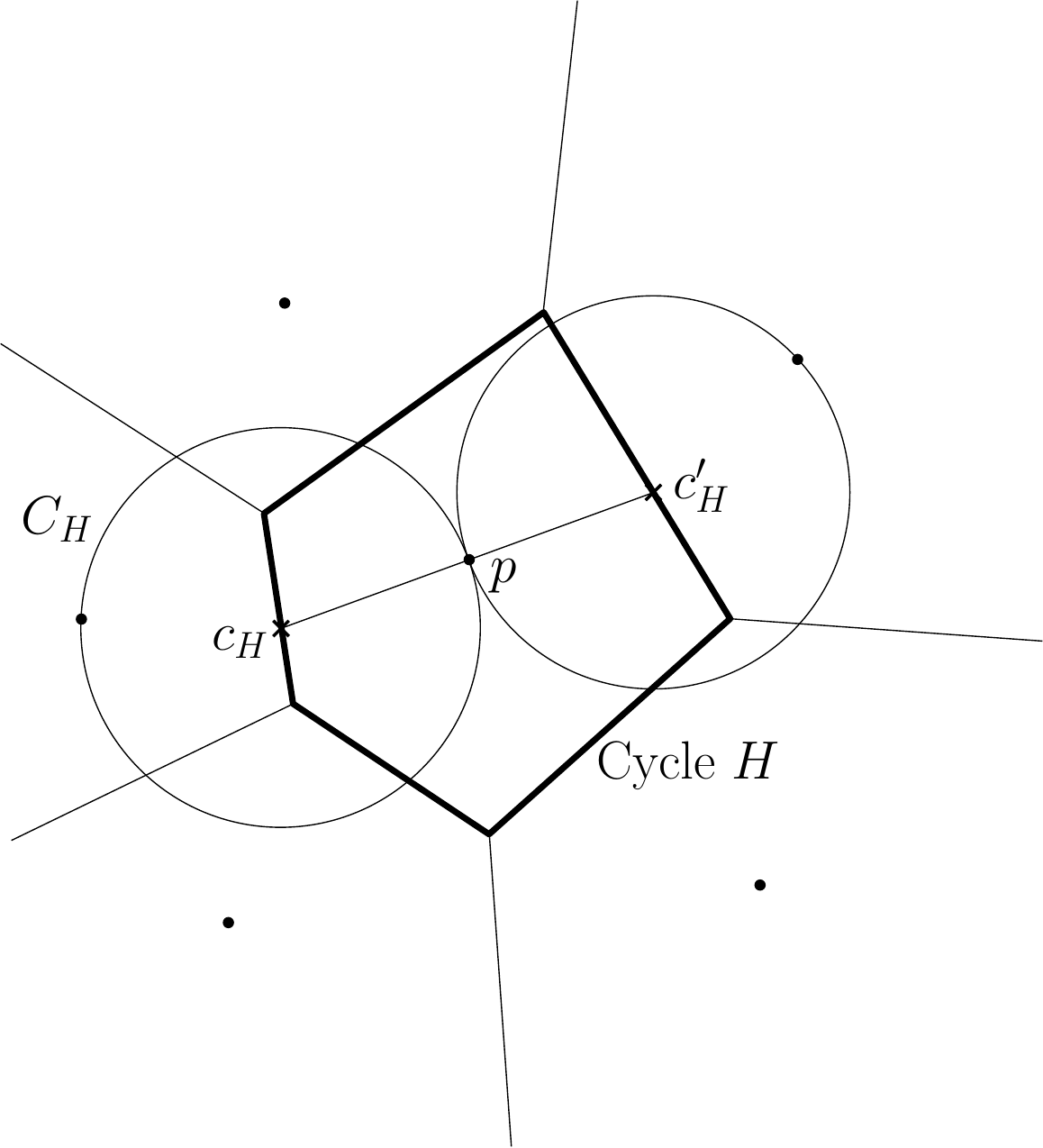}
\caption{Illustration of Lemma~\ref{lem:cycle}; the edges in $H$ are
shown darker} 
\label{fig:cycle}
\vspace{-0.1in}
\end{figure}
\qed
\end{proof}

\begin{lemma}[Unique Path Lemma]\label{lem:onepath}
Let $C$ and $C'$ be any two distinct but overlapping MECs with center
at $c$ and $c'$ respectively. There is exactly one path from $c$ to
$c'$ along the Voronoi edges such that every MEC centered on that path
encloses $C\cap C'$.
\end{lemma}

\begin{proof}
The structure of the proof is as follows. We provide a procedure that 
constructs a path $\Pi(c,c')$ from $c$ to $c'$ along the Voronoi
edges, and ensure that every MEC centered on that path encloses $C
\cap C'$. As a consequence of Lemma \ref{lem:cycle}, the path does not
form an intermediate cycle and terminates at $c'$. Finally, we again
use Lemma~\ref{lem:cycle} to show that no path $\P$, other than
$\Pi(c,c')$, exists between $c$ and $c'$ such that every MEC centered
on $\P$ contains $C \cap C'$. Throughout this proof, we closely follow
Figure~\ref{fig:nsp} in order to keep the arguments intuitive. To keep
arguments simple, we assume that $c$ and $c'$ are Voronoi vertices.
When $c$ and $c'$ are not Voronoi vertices, then also the same
argument follows. 

Let $\alpha$ be the number of points in $P$ that $C$  touches. These
$\alpha$ points partition  $C$ into $\alpha$ arcs. The degree of
the corresponding Voronoi vertex $c$ is also $\alpha$ because each
adjacent pair of points from $P$ that lie on the boundary of $C$ will
induce a Voronoi edge incident on $c$ and vice versa. These Voronoi
edges and their corresponding arcs are denoted by $e^j_C$ and $s^j_C$,
for $1 \le j \le \alpha$.

Consider the other  MEC $C'$ ($\neq C$ and centered at a vertex $c'$)
that overlaps with $C$. $C'$ intersects $C$ at two points $t_1$ and
$t_2$; both $t_1$ and $t_2$ must lie in one of the $\alpha$ arcs of
$C$ (due to the emptiness of $C'$). Let us name this arc by $s^j_C$.
Consider the edge $e^j_C = (c,c_2)$ that corresponds to the arc
$s^j_C$. The other end of $e^j_C$, i.e., the vertex $c_2$, is called
the {\em next step from $c$ toward $c'$} and denote it as
$\mathbf{ns}(c,c')$. Consider the following code that generates a path
denoted by $\mathbf{\Pi}(c,c')$:

\vspace{-0.1in}

\begin{algorithm}[h!]
\caption{$\Pi(c,c')$ Computation}
\label{alg:pi}
\begin{algorithmic}[1]
\STATE $\mathbf{\Pi}(c,c') \leftarrow (c)$
%\COMMENT{Initialize $\mathbf{nsp}(c,c')$.}
\STATE $\mathbf{next} \leftarrow c$
\REPEAT
\STATE $\mathbf{next} \leftarrow \mathbf{ns}(\mathbf{next},c')$
\STATE Append $\mathbf{next}$ to $\mathbf{\Pi}(c,c')$
\UNTIL{$\mathbf{next}$ equals $c'$}
\COMMENT{This is the only terminating condition.}
\end{algorithmic}
\end{algorithm}

\begin{figure}[h]
\centering
\includegraphics[width=0.5\columnwidth]{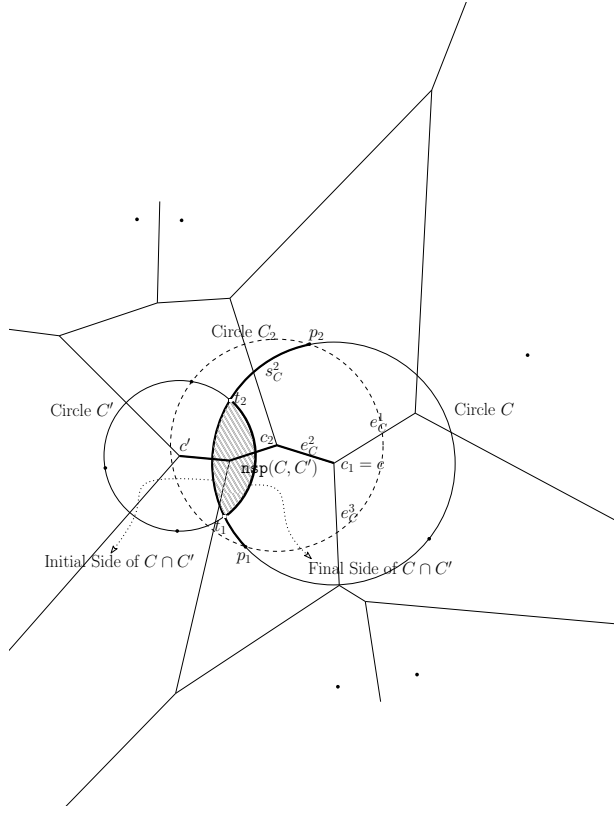}
\caption{Illustration of $\Pi(c,c')$.}
\label{fig:nsp}\vspace{-0.1in}
\end{figure}

Let $\mathbf{\Pi}(c,c')$ as constructed above be $(c_1=c, c_2, \ldots,
c_i, c_{i+1}, \ldots, c')$. Let $C_2$ denote the
MEC centered at $c_2$. If $C_2$ is the circle $C'$, then the procedure
terminates and, as required, every MEC in the edge $(c,c_2)$ 
encloses $C \cap C_2 = C \cap C'$.  

Therefore, consider the case where $C_2$ is not $C'$. Let $p_1, p_2
\in P$ be the points at which $C$ and $C_2$ intersect; $p_1,p_2$ are
the end points of the arc $s^j_C$ that defines the next step move
toward $c'$ (in Figure \ref{fig:nsp}, $j$ is $2$). Therefore, by
definition, $t_1$ and $t_2$ lie on the arc $s^j_C$. Notice that $C
\cap C'$ (shown shaded in Figure~\ref{fig:nsp}) is shaped like a rugby
ball with $t_1$ and $t_2$ at its end-points. One side of $C \cap C'$
(called the {\it initial side}) is in $C$ and the other side 
(called the {\it final side}) is in $C'$. Clearly, $t_1$ and
$t_2$ are inside (or on the boundary of) every MEC centered on the
edge $e^j_C$. Otherwise, as we go from $C$ to $C_2$, there will be a
circle that touches the final side of $C \cap C'$, but that would mean
that we have either 
\vspace{-0.1in}
\begin{itemize}
\item reached $C'$, which contradicts our assumption that $C_2$ is not
$C'$, 
\item or found a MEC that contains $C'$, which contradicts the fact
that $C'$ is itself an MEC.
\end{itemize}
\vspace{-0.1in}

We now make two observations: (i) $C$ touches the initial side, but
(ii) no other MECs centered on $e^j_C$ (and $C_2$ in particular)
touches the final side Suppose (for the sake of contradiction) that
there is a MEC $C^*$ centered on $e^j_C$ that touches the final side
of $C \cap C'$ at, say, some point $t^*$. It is easy to see that $C^*$
will contain $C'$ because it touches $t^*$ on $C'$ and contains $t_1$
and $t_2$, which are also on $C'$ --- this is a contradiction due to
the observations (i) and (ii), stated above. Thus, it is clear that $C
\cap C'$ is properly contained within $C \cap C_2$.  

Consider two adjacent vertices $c_i$ and $c_{i+1}$ along $\Pi(c,c')$
with MECs $C_i$ and $C_{i+1}$ centered on them, respectively. The
above argument can be easily extended to give us the following: 
$$C_i \cap C' \subset C_{i+1} \cap C'.$$ Therefore, we can conclude
that every MEC along $\Pi(c,c')$ encloses $C \cap C'$. Given
Lemma~\ref{lem:cycle}, we can also conclude that $\Pi(c,c')$ does not
form a cycle. The only stopping condition is when we actually reach
$c'$, so $\Pi(c,c')$ terminates at $c'$ in $O(n)$ steps. Therefore, 
$\Pi(c,c')$ fulfils our requirements.

To complete the proof of this lemma, we must show that $\Pi(c,c')$ is
the only required path. For the sake of contradiction, assume that
there is another path $\Pi'$ such that every MEC centered on $\Pi'$
contains $C \cap C'$. Then, there are two distinct paths from $c$ to
$c'$ such that every MEC centered on both paths overlapped with $C
\cap C'$. Clearly, there must be a cycle when the two paths are
combined. From Lemma~\ref{lem:cycle}, we know that there are pairs of
MECs in the cycle that will not overlap each other.  This is a
contradiction, thus proving that $\Pi(c_1,c_2)$ is the only required
path and concluding the proof of the lemma. \qed
\end{proof}

\begin{figure}
\centering
\includegraphics[width=0.50\textwidth]{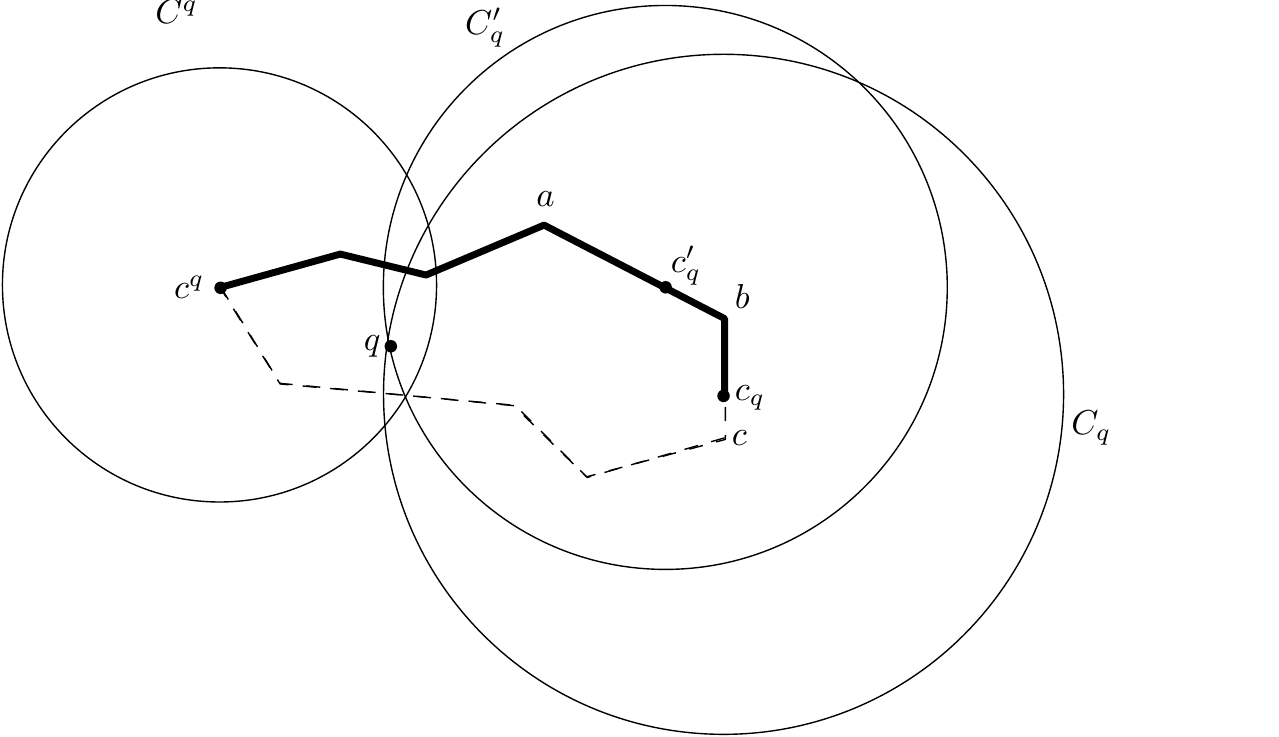}
\caption{Illustrates the usefulness of the Unique Path Lemma.}
\label{fig:ThoughtfulExample}\vspace{-0.1in}
\end{figure}

To illustrate the usefulness of the {\it unique path lemma}, consider
the example depicted in Figure \ref{fig:ThoughtfulExample}. Let $C^q$
(centered at $c^q$) be the MEC returned by the LCQ data structure when
queried with $q$. Suppose that $C_q$ (centered on $c_q$) is the
largest empty circle containing $q$. Suppose Procedure~\ref{alg:query}
reports $C'_q$ centered on $c'_q$, which lies on an overlapping edge
$(a,b)$. In terms of size, let $C^q < C'_q < C_q$. Such a behavior
will render our algorithm incorrect. However, this incorrect behavior
is only possible when $c_q$ does not lie on an overlapping edge with
respect to $c^q$, and therefore, Procedure~\ref{alg:query} may never
find it. This may happen in two possible situations: 
\begin{description}
\item[Situation 1:] {\bf The path shown in thick continuous segments
is $\Pi(c^q,c_q)$.} Since Procedure~\ref{alg:query} reported $c_q'$,
the edge $(a,b)$ is  an overlapping edge, the MEC at $b$ must
be larger than $C^q$. However, the LCQ data structure did not report
the MEC at $b$, but rather reported $C^q$.  While $q \in C^q$ and $q
\in C_q$, $q$ is not contained within the MEC at $b$. From the {\it
unique path lemma}, clearly, $b$ cannot be in $\Pi(c^q,c_q)$. Thus
such a situation is impossible. 
\item[Situation 2:] {\bf Some other path that does not contain
$c_q'$ (shown using dashed line segments) is $\Pi(c^q,c_q)$.} This
case is possible and the edge $(c,b)$, shown in Figure
\ref{fig:ThoughtfulExample}, is an overlapping edge through that
dashed path. Therefore, the edge $(c,b)$ will be considered by
Procedure~\ref{alg:query} and $C_q$ will be reported correctly. 
\end{description}

\begin{lemma}\label{lem:PointSetCorrectness}
If a set $P$ of points from $\mathbb{R}^2$ are preprocessed by
Procedure~\ref{alg:preprocessing}, then Procedure~\ref{alg:query},
when invoked with a query point $q$, correctly reports the largest
circle that contains $q$ but is devoid of points from $P$. 
\end{lemma}
\begin{proof}
Let $C_q$ be the largest circle that contains $q$ but is devoid of
points from $P$. In short, Procedure~\ref{alg:query} must report
$C_q$. If $C_q$ is centered on a Voronoi vertex, then, clearly,
Procedure~\ref{alg:query} reports it correctly.  

We now show that if $q$ is contained by a MEC centered at an
artificial vertex $v$, then also the algorithm reports the correct
$C_q$.
Let $e$ be the half line edge containing $v$.
Procedure~\ref{alg:query} only searches edge $e$, which, we claim is
sufficient. Suppose for the sake of contradiction that $C_q$ is
centered on some other edge $e'$.  Clearly, $e'$ cannot be a bounded
Voronoi edge, because the MEC at $v$ is larger than MECs centered on
bounded Voronoi edges. Recall that our construction of artificial
vertices ensures that no two MECs centered on artificial vertices will
overlap inside the convex hull of $P$. Therefore, $e'$ cannot be an
edge that is a half line either, because, the MEC at $v$ contains $q$,
so the MEC centered at the artificial vertex on $e'$  cannot contain
$q$. Therefore, such an $e'$ cannot exist and we can conclude that the
algorithm correctly reports the largest MEC  centered on some point in
$e$ that contains $q$ as $C_q$.

For the rest of the proof, we assume that $C_q$ is not centered on a
Voronoi vertex and $q$ is not enclosed by any MEC centered at an
artificial vertex. Recall from Procedure~\ref{alg:query} that $C^q$
is the largest MEC that is centered on a vertex and contains $q$.
Clearly, $C^q \cap C_q \ne \emptyset$ as it at least contains $q$.
From Lemma~\ref{lem:onepath}, there is exactly one path $\Pi(c^q,
c_q)$ from the center of $C^q$ to the center of $C_q$ such that every
circle in that path contains $C^q \cap C_q$. Clearly,  MECs centered
at vertices in $\Pi(c^q,c_q)$ other than the centers of $C^q$
and $C_q$ are smaller than $C^q$; otherwise, the {\tt LCQ} data
structure would not have chosen $C^q$. However, $C_q$ is larger than
$C^q$. Consider the two vertices connected by the edge that contains
the center of $C_q$. The MEC centered on one of them must be strictly
larger than $C^q$, while the MEC on the other must be strictly smaller
than $C^q$. Therefore, it is easy to see that $C_q$ is centered on an
overlapping edge. Since the algorithm searches through all overlapping
edges, it will find and report $C_q$ correctly. \qed
\end{proof}
Lemma~\ref{lem:PointSetBounds} coupled with the line sweep method of
implementing the {\tt LCQ} data structure and
Lemma~\ref{lem:PointSetCorrectness} immediately lead to the following
theorem. 

\begin{theorem}\label{thm:PointSetResult}
Given a set $P$ of points in $\mathbb{R}^2$, we can preprocess $P$ in
$O(n^2 \log n)$ time and $O(n^2)$ space so that the resulting data
structure can be queried for the largest empty circle containing the
query point $q$ in $O(\log n)$ time.  
\end{theorem}

\section{QMER problem}\label{sec:qmer}
The input consists of a set of points $P=\{p_1, p_2, \ldots, p_n\}$ 
in a rectangular region $\cal A$. An axis-parallel rectangle inside
$\cal A$ is said to be an {\it empty rectangle} if it does not contain
any point of $P$. An empty rectangle is called {\it maximal empty
rectangle} (MER) if no other empty rectangle in $\cal A$ properly
contains it. Here our objective is to preprocess $P$ such that given
any arbitrary query point $q \in  \cal A$, the largest area rectangle
$L(q)$ inside $\cal A$ that does not contain any point in $P$ but
contains $q$, can be reported efficiently.  First observe that 
$L(q)$ is an MER. Let $M$ denote the set of all possible MERs in
$\cal A$. In \cite{NHL}, it is shown that  $|M|=\Theta(n^2)$ in the
worst case.

In the preprocessing phase, we partition  $\cal A$ into a set $C$ of
cells, such that for every point $q$ inside a cell $c \in C$, the
largest MER containing $q$ is the same. This is achieved by drawing 
horizontal and a vertical lines through each point in $P$. This
splits $\cal A$ into the set $C$ of $O(n^2)$ cells. Observe that for
any cell $c \in C$ and any MER $M_\alpha\in M$, either $c\cap
M_\alpha=c$ or $c\cap M_\alpha=\phi$. Therefore, if the query point
$q$ lies inside  $c \in C$, we need to report the largest MER
containing $c$. We choose a representative point inside each cell $c
\in C$. Let $Q$ be the set of representative points. We compute all
the MERs using the algorithm in \cite{NHL}, and sort them with respect
to their area. We also construct an augmented dynamic range tree
${\T}$ with the points in $Q$ in $O(n^2\log n)$ time and space
\cite{MN}.  Next, we process the members in $M$ in order. For each
$M_\alpha \in M$, we identify the set $Q_\alpha$ of points in $Q$ that
are inside $M_\alpha$. We store a pointer to $M_\alpha$ along with
each point in $Q_\alpha$ and then delete $Q_\alpha$ from ${\T}$. This
step takes $O(L+\log n)$ time~\cite{MN}, where $L$ is the number of
points inside $M_\alpha$. After processing all the MERs in $M$, we
have stored a pointer to the largest MER along with each $q \in Q$.
Thus, we have the following theorem:

\begin{theorem}
A set of $n$ points in a rectangle $A$ can be preprocessed  in
$O(n^2\log n)$ time and space so that the largest empty rectangle
query containing the query point can be answered in $O(\log n)$ time.
\end{theorem}
Moreover, it is not hard to see that we can construct examples, where
there are $\Omega(n^2)$ cells, so that the MER containing each of
these cells is combinatorially different, i.e., the boundary of any of
the two MERs are not incident to the same set of points in $P$.  This 
suggests that in order to answer queries in polylogarithmic time, we
need to somehow store $\Omega(n^2)$ cells in a data structure, and
hence it is very unlikely to improve the cost of preprocessing in
order to maintain $O(\log n)$ query time. 

\section{Future Work} 
Our focus in this paper has been in terms of understanding which
problems can be solved within subquadratic preprocessing time, while
maintaining the polylogarithmic query time.   At this stage the
central problem here is to  understand whether the preprocessing time
for the QMEC problem for the point set case can be tightened to a
subquadratic bound.
 A possible lead for improvement is as follows.
If we were to insert $q$ into the
set of points $P$, and compute the Delaunay triangulation of the new
set, then the query circle, $C_q$, is a circumcircle of one of the
triangles $t$ incident to $q$ - in fact the angle subtended at $q$ in 
$t$ will be greater than $\pi/2$.  The running time for inserting $q$
in the
Delaunay triangulation of $P$ is proportional to the degree of $q$. 
One may be tempted to
use a (randomized) incremental algorithm for constructing Delaunay
triangulation, but there are cases in which the degree of $q$ can be
linear. This approach  may lead to a practical and a simpler way to
handle these types of queries, but this remains to be seen.  
Alternatively, one may look into the divide and conquer algorithms for
computing the Voronoi diagram, and  see whether in the ``merge step'',
the maximal empty circles can be maintained. Alternatively, one may try to use the planar separator theorem to partition the Voronoi diagram, recursively, and for the separator vertices (point), build an appropriate structure, so that preprocessing can be performed in subquadratic time, and the queries can be answered in sublinear time.

It will be desirable to improve the preprocessing cost in the case of
simple polygons. A real challenge will be to  match the complexity in
this case to exactly that of the convex polygon case.

While we have studied a few canonical problems, there are several
other variants that are as yet untouched. We can also ask similar
questions on multidimensional geometric sets, but we suspect that the
curse of dimensionality might restrict us to approximations.

\noindent {\bf Acknowledgments}: We are grateful to Samir Datta and
Vijay Natarajan for their helpful suggestions and ideas.
We are also thankful to Subir Ghosh for providing the environment to
carry out this work.

\end{document}